\documentclass[11pt]{article}

\usepackage{xspace}
\usepackage{verbatim}
\usepackage[dvipsnames]{xcolor}
\usepackage{amsmath}
\usepackage{amssymb}
\usepackage{amsthm}
\usepackage{bbm}
\usepackage[showonlyrefs]{mathtools}
\usepackage{mathstyle}
\usepackage{empheq}
\usepackage{amstext,amssymb,amsfonts}
\usepackage{fullpage}
\usepackage{nicefrac}
\usepackage{bm}

\usepackage[linesnumbered,ruled,vlined]{algorithm2e}

\usepackage{textcomp,setspace}

\usepackage{nameref}
\usepackage[pagebackref,colorlinks,linkcolor=blue,filecolor = blue,citecolor = blue, urlcolor = blue, hyperfootnotes=false]{hyperref}
\usepackage{color}
\usepackage[capitalize, noabbrev, nameinlink]{cleveref}

\usepackage{thmtools}
\usepackage{thm-restate}

\declaretheorem[within=section]{theorem}
\declaretheorem[sibling=theorem]{corollary}

\declaretheorem[sibling=theorem]{lemma}

\declaretheorem[sibling=theorem]{definition}

\declaretheorem[sibling=theorem]{remark}

\declaretheorem[sibling=theorem]{assumption}
\declaretheorem[sibling=theorem]{example}

\crefname{assumption}{Assumption}{Assumptions}
\newcounter{termcounter}
\renewcommand{\thetermcounter}{\Alph{termcounter}}
\crefname{term}{term}{terms}
\creflabelformat{term}{\textup{(#2#1#3)}}

\makeatletter
\def\term{\@ifnextchar[\term@optarg\term@noarg}
\def\term@optarg[#1]#2{%
  \textup{(#1)}%
  \def\@currentlabel{#1}%
  \def\cref@currentlabel{[][2147483647][]#1}%
  \cref@label[term]{#2}}
\def\term@noarg#1{%
  \refstepcounter{termcounter}%
  \textup{(\thetermcounter)}%
  \cref@label[term]{#1}}
\makeatother

\newcommand{\ignore}[1]{}

\newcommand{\eval}{\mbox{$\mathsf{EVAL}$}~}
\newcommand{\samp}{\mbox{$\mathsf{SAMP}$}~}

\newcommand{\poly}{\mathrm{poly}}

\newcommand{\KL}{\mathrm{KL}}

\newcommand{\brac}[1]{[#1 ]}
\newcommand{\Brac}[1]{\left[#1 \right]}
\newcommand{\set}[1]{\{#1\}}

\definecolor{DSred}{rgb}{1,0,0}

\renewcommand{\leq}{\leqslant}
\renewcommand{\geq}{\geqslant}
\renewcommand{\ge}{\geqslant}
\renewcommand{\le}{\leqslant}
\renewcommand{\epsilon}{\varepsilon}
\newcommand{\eps}{\epsilon}

\newcommand{\R}{\mathbb{R}}

\newcommand{\cA}{\mathcal A}

\newcommand{\cC}{\mathcal C}
\newcommand{\cD}{\mathcal D}

\newcommand{\cL}{\mathcal L}

\newcommand{\cP}{\mathcal P}
\newcommand{\cQ}{\mathcal Q}

\newcommand{\cS}{\mathcal S}

\newcommand{\Esymb}{{\bf E}}

\newcommand{\Psymb}{{\bf Pr}}

\DeclareMathOperator*{\E}{\Esymb}

\DeclareMathOperator*{\ProbOp}{\Psymb}

\renewcommand{\Pr}{\ProbOp}

\newcommand{\ex}[1]{\E\brac{#1}}

\newcommand{\argEx}[2]{\E_{#1}\Brac{#2}}

\newcommand{\dtv}{\mbox{${d}_{\mathrm TV}$}}

\DeclareMathOperator*{\Pa}{{\mathsf{Pa}}}

\def\notes{1}
 \newcommand{\anote}[1]{\ifnum\notes=1{\mnote{Arnab: #1}}\fi}
 \newcommand{\snote}[1]{\ifnum\notes=1{\mnote{Sutanu: #1}}\fi}
 \newcommand{\vnote}[1]{\ifnum\notes=1{\mnote{Vinod: #1}}\fi}
\newcommand{\knote}[1]{\ifnum\notes=1{\mnote{Kuldeep: #1}}\fi}

\newcommand{\cdo}[0]{\mathsf{do}}

\newcommand{\cbn}{\textsc{CBN}} 

\title{Efficient Distance Approximation for Structured High-Dimensional Distributions via Learning\footnote{Author names are in alphabetical order}}
\author{
Arnab Bhattacharyya\thanks{National University of Singapore. Supported in part by Start-up Grant WBS R252000A33133 and an Amazon Research Award.}\\
arnabb@nus.edu.sg
\and
Sutanu Gayen\thanks{National University of Singapore. Supported in part by AB's Start-up Grant WBS R252000A33133.}\\
sutanugayen@gmail.com
\and
Kuldeep S. Meel\thanks{National University of Singapore.}\\meel@comp.nus.edu.sg
\and
N.~V.~Vinodchandran\thanks{University of Nebraska, Lincoln. Research mostly conducted while visiting National University of Singapore.}\\vinod@unl.edu
}

\begin{document}

\maketitle
\begin{abstract}
We design efficient distance approximation algorithms for several classes of structured high-dimensional distributions. Specifically, we show algorithms for the following problems:
\begin{itemize}
\item Given sample access to two Bayesian networks $P_1$ and $P_2$ over known directed acyclic graphs $G_1$ and $G_2$ having $n$ nodes and bounded in-degree, approximate $\dtv(P_1,P_2)$ to within additive error $\eps$ using $\poly(n,\eps)$ samples and time
\item Given sample access to two ferromagnetic Ising models $P_1$ and $P_2$ on $n$ variables with bounded width, approximate $\dtv(P_1, P_2)$ to within additive error $\eps$ using $\poly(n,\eps)$ samples and time
\item Given sample access to two $n$-dimensional gaussians $P_1$ and $P_2$, approximate $\dtv(P_1, P_2)$ to within additive error $\eps$ using $\poly(n,\eps)$ samples and time
\item
Given access to observations from two causal models $P$ and $Q$ on $n$ variables that are defined over known causal graphs, approximate $\dtv(P_a, Q_a)$ to within additive error $\eps$ using $\poly(n,\eps)$ samples, where $P_a$ and $Q_a$ are the interventional distributions obtained by the intervention $\cdo(A=a)$ on $P$ and $Q$ respectively for a particular variable $A$
\end{itemize}
Our results are the first efficient distance approximation algorithms for these well-studied problems. They are derived using a simple and general connection to distribution learning algorithms. The distance approximation algorithms imply new efficient algorithms for {\em tolerant} testing of closeness of the above-mentioned structured high-dimensional distributions.  
\end{abstract}

\newpage


\section{Introduction}
A fundamental challenge in statistics and computer science is to devise hypothesis tests that use a small number of samples. A classic problem of this type is {\em identity testing} (or, {\em goodness-of-fit testing}): given samples from an unknown distribution $P$ over a domain $\cS$, does $P$ equal a specific reference distribution $Q$? A sequence of works \cite{Pan08, BatuFRSW13, Valiant:2014:AIP:2706700.2707449, ChanDVV14} in the property testing literature has pinned down the finite sample complexity of this problem. It is known that with $O(|\cS|^{1/2}\eps^{-2})$ samples from $P$, one can,  with probability at least $2/3$, distinguish whether $P=Q$ or whether $\dtv(P,Q)>\eps$; also, $\Omega(|\cS|^{1/2}\eps^{-2})$ samples are necessary for this task. A related problem is {\em closeness testing} (or, {\em two-sample testing}): given samples from two unknown distributions $P$ and $Q$ over $\cS$, does $P=Q$? Here, it is known that $\Theta(|\cS|^{2/3}\eps^{-4/3} + |\cS|^{1/2}\eps^{-2})$ samples are necessary and sufficient to distinguish $P=Q$ from $\dtv(P,Q)>\eps$ with probability at least $2/3$. The corresponding algorithms for both identity and closeness testing run in time polynomial in $|\cS|$ and $\eps^{-1}$. 

However, in order to solve these testing problems in many real-life settings, there are two issues that need to be surmounted.
\begin{itemize}
\item
\textbf{High dimensions:}
In typical applications, the data is described using a huge number of (possibly redundant) features; thus, each item in the dataset is represented as a point in a high-dimensional space. If $\cS = \Sigma^n$, then from the results quoted above, identity testing or closeness testing for arbitrary probability distributions over $\cS$ requires $2^{\Omega(n)}$ many samples, which is clearly unrealistic. Hence, we need to restrict the class of input distributions.
\item
\textbf{Approximation:} A high-dimensional distribution requires a large number of parameters to be specified. So, for identity testing, it is unlikely that we can ever hypothesize a reference distribution $Q$ such that it exactly equals the data distribution $p$. Similarly, for closeness testing, two data distributions $P$ and $Q$ are most likely not exactly equal. Hence, we would like to design {\em tolerant} testers for identity and closeness that distinguish between $\dtv(P,Q)\leq \eps_1$ and $\dtv(P,Q)>\eps_2$ where $\eps_1$ and $\eps_2$ are user-supplied parameters.
\end{itemize}

In this work, we design sample- and time-efficient tolerant identity and closeness testers for natural classes of distributions over $\Sigma^n$. More precisely, we focus on {\em distance approximation} algorithms:
\begin{definition}
Let $\cD_1, \cD_2$ be two families of distributions over $\Sigma^n$.
A {\em distance approximation algorithm for $(\cD_1, \cD_2)$} is a randomized algorithm $\cA$ which takes as input $\eps \in (0,1)$, and sample access to two unknown distributions $P \in \cD_1, Q \in \cD_2$. The algorithm $\cA$ returns as output a value $\gamma \in [0,1]$ such that, with probability at least $2/3$:
$$\gamma - \eps \leq \dtv(P,Q) \leq \gamma + \eps.$$ 
If $\cD_1=\cD_2=\cD$, then we refer to such an algorithm as a {\em distance approximation algorithm for $\cD$}. 
\end{definition}
\begin{remark}
{The success probability can be amplified to $1-\delta$ by taking the median of $O(\log \delta^{-1})$ independent repetitions of the algorithm with success probability $2/3$.}
\end{remark}

The distance approximation problem and the tolerant testing problem are equivalent in the setting we consider. A distance approximation algorithm for $(\cD_1, \cD_2)$ immediately gives a tolerant closeness  testing algorithm for  two input distributions $P \in \cD_1$ and $Q \in \cD_2$  with the same asymptotic sample and time complexity bounds. Also a tolerant closeness testing algorithm for distributions in $\cD_1$ and $\cD_2$ gives a distance approximation algorithm for $(\cD_1, \cD_2)$, although with slightly worse sample and time complexity bounds (resulting from a binary search approach). Indeed this connection was explored in the property testing setting in~\cite{DBLP:journals/jcss/ParnasRR06} which established a general translation result. Thus, in the rest of this paper we will focus on the distance approximation problem and the results translate to appropriate tolerant testing problems. The bounds on the sample and time complexity will be phrased in terms of the description lengths of $\cD_1$ and $\cD_2$.

\ignore{\subsection{Organization}
The rest of the paper is organized as follows. In \cref{sec:new}, we summarize our results on distance approximation for structured high dimensional distributions. We also summarize prior related work at the end of this section. In \cref{sec:approx},
we present our main algorithm that approximates the distance between two distributions which can be sampled and approximately evaluated. In \cref{sec:bn}, we apply our technique to approximate the distance between two  high-dimensional Bayesian networks. \ignore{Due to space limitations, we defer the details 
of the rest of the results to the Appendix. }}

\section{New Results}\label{sec:new}

We design new sample and time efficient distance approximation algorithms for several well-studied families of high-dimensional distributions given sample access. We accomplish this by prescribing  a general strategy for designing distance approximation algorithms. In particular, we first design an algorithm to approximate the distance between a pairs of distributions. However, this algorithm needs  both sample access and an approximate evaluation oracle. We crucially observe that a learning algorithm 
that outputs a representation of the unknown distribution given sample access, can often efficiently simulate 
the approximation oracle. Thus the final algorithm only needs sample access. This general strategy  coupled with appropriate learning algorithms, leads to a number of new distance approximation algorithms (and hence new tolerant testers) for well-studied families of high-dimensional probability distributions. 

\subsection{Distance Approximation from \eval Approximators}
Given a family of distributions $\cD$, a learning algorithm for $\cD$  is an algorithm $\cL$ that on input $\eps \in (0,1)$ and sample access to a distribution $P$ promised to be in $\cD$, returns the description of a distribution $\hat{P}$ such that with probability at least $2/3$, $\dtv(P, \hat{P}) \leq \eps$. It turns out that for many natural distribution families $\cD$ over $\Sigma^n$, one can easily modify known learning algorithms for $\cD$ to efficiently output not just a description of $\hat{P}$ but the value of $\hat{P}(x)\coloneqq \Pr_{X \sim \hat{P}}[X=x]$ for any $x \in \Sigma^n$. More precisely, they yield what we call {\em \eval approximators}:
\begin{definition}\label{def:oracle}
Let $P$ be a distribution over a finite set $U$. A function $E_P:U \to [0,1]$ is a {\em $(\beta,\gamma)$- \eval approximator} for $P$ if there exists a distribution $\hat{P}$ over $U$ such that
\begin{itemize}
\item $\dtv(P,\hat{P})\le \beta$
\item $\forall x \in U$, $(1-\gamma)\cdot\hat{P}(x) \le E_P(x) \le (1+\gamma)\cdot\hat{P}(x)$
\end{itemize}
\end{definition}
\noindent Typically, the learning algorithm outputs parameters that describe $\hat{P}$, and then $\hat{P}(x)$ can be computed (or approximated) efficiently in terms of these parameters.

\begin{example}\label{ex:prod}
Suppose $\cD$ is the family of product distributions on $\{0,1\}^n$. That is, any $P \in \cD$ can be described in terms of $n$ parameters $p_1, \dots, p_n$ where each $p_i$ is the probability of the $i$'th coordinate being $1$. It is folklore that there is a learning algorithm which gets $O(n\eps^{-2})$ samples from $P$ and returns the parameters $\hat{p}_1, \dots, \hat{p}_n$ of a product distribution $\hat{P}$ satisfying $\dtv(P,\hat{P})\leq \eps$ with probability $2/3$. It is clear that given $\hat{p}_1, \dots, \hat{p}_n$, we can compute $\hat{P}(x)$ for any $x \in \{0,1\}^n$ in linear time as:
\[
\hat{P}(x) = \prod_{i=1}^n \left(x_i \cdot \hat{p}_i + (1-x_i)\cdot (1-\hat{p}_i)\right)
\]
Thus, there is an algorithm that takes as input sample access to any product distribution $P$, has sample and time complexity $O(n\eps^{-2})$, and returns a circuit implementing an $(\eps, 0)$-\eval approximator for $P$. Moreover, any call to the circuit returns in $O(n)$ time.
\end{example}

We establish the following link between \eval approximators and distance approximation.
\begin{theorem}\label{thm:maininf}
Suppose we have sample access to distributions $P$ and $Q$ over a finite set. Also, suppose we have access to $(\epsilon, \epsilon)$-\eval approximators for $P$ and $Q$. Then, with probability at least 2/3, $\dtv(P,Q)$ can be approximated to within $O(\epsilon)$ additive error using $O(\epsilon^{-2})$ samples from $P$ and $O(\epsilon^{-2})$ calls to the two \eval approximators.
\end{theorem}

Thus, in the context of \cref{ex:prod}, the above theorem immediately implies a distance approximation algorithm for product distributions using $O(n \eps^{-2})$ samples and time. \cref{thm:maininf} extends the work of Canonne and Rubinfeld~\cite{DBLP:conf/icalp/CanonneR14} who considered the setting $\beta = \gamma = 0$. We discuss the relation to prior work in \cref{sec:prior}.

 
\subsection{Bayesian Networks}A standard way to model structured high-dimensional distributions is through {\em Bayesian networks}. A Bayesian network describes how a collection of random variables can be generated one-at-a-time in a directed fashion, and they have been used to model beliefs in a wide variety of domains (see \cite{JN07, KF09} for many pointers to the literature). Formally, a probability distribution $P$ over $n$ variables $X_1, \dots, X_n \in \Sigma$ is said to be a {\em Bayesian network on a directed acyclic graph $G$} with $n$ nodes if\footnote{We use the notation $X_S$ to denote $\{X_i : i \in S\}$ for a set $S \subseteq [n]$.} for every $i \in [n]$, $X_i$ is conditionally independent of $X_{\text{non-descendants}(i)}$ given $X_{\text{parents}(i)}$. Equivalently, $P$ admits the factorization:
\begin{equation}\label{eqn:bnfactor}
P(x) \coloneqq \Pr_{X \sim P}[X=x]= \prod_{i=1}^n \Pr_{X\sim P}[X_i = x_i \mid \forall j \in {\rm parents}(i), X_j = x_j] \qquad \text{for all } x \in \Sigma^n
\end{equation}
For example, product distributions are Bayesian networks on the empty graph.

Invoking our framework of distance approximation via \eval approximators on Bayesian networks, we obtain the following:
\begin{restatable}{theorem}{bnmain}\label{thm:bn-main}
Suppose $G_1$ and $G_2$ are two DAGs on $n$ vertices with in-degree at most $d$. Let $\cD_1$ and $\cD_2$ be the family of Bayesian networks on $G_1$ and $G_2$ respectively. Then, there is a distance approximation algorithm for $(\cD_1, \cD_2)$  that gets $m=\tilde{O}(|\Sigma|^{d+1} n \eps^{-2})$ samples and runs in $O(|\Sigma|^{d+1} mn)$ time.
\end{restatable} 

We design a learning algorithm for Bayesian networks  on a known DAG  $G$ that uses $\tilde{O}(n \epsilon^{-2} |\Sigma|^{d+1})$ samples where $d$ is the maximum in-degree. 
It returns another Bayesian network $\hat{P}$ on $G$, described in terms of the conditional probability distributions $X_i \mid x_{\text{parents}(i)}$ for all $i \in [n]$ and all settings of $x_{\text{parents}(i)} \in \Sigma^{\deg(i)}$. Given these conditional probability distributions, we can easily obtain $\hat{P}(x)$ for any $x$, and hence, an $(\eps,0)$-\eval approximator for $P$, by using (\ref{eqn:bnfactor}).  \cref{thm:bn-main} then follows from \cref{thm:maininf}.

\cref{thm:bn-main} extends the works of Daskalakis et al.~\cite{DBLP:conf/colt/DaskalakisP17} and Canonne et al.~\cite{DBLP:conf/colt/CanonneDKS17}  who designed efficient  {\em non-tolerant} identity and closeness testers for Bayesian networks. Their arguments appear to be inadequate to design tolerant testers. In addition, their results for general Bayesian networks were restricted to the case when $G_1 = G_2$. \cref{thm:bn-main} immediately gives efficient {\em tolerant} identity and closeness testers for Bayesian networks even when $G_1\neq G_2$. Canonne et al.~\cite{DBLP:conf/colt/CanonneDKS17} obtain better sample complexity but they make certain {\em balancedness} assumption on each conditional probability distribution. Without such assumptions, the sample complexity of our algorithm is optimal.

\subsection{Ising Models}\label{sec:intro-ising}
 Another widely studied model of high-dimensional distributions is the {\em Ising model}. It was originally introduced in statistical physics as a way to study spin systems (\cite{Isi25}) but has since emerged as a versatile framework to study other systems with pairwise interactions, e.g., social networks (\cite{MS10}), learning in coordination games (\cite{Ell93}), phylogeny trees in evolution (\cite{Ney71, Far73, Cav78}) and image models for computer vision (\cite{GG86}). Formally, a distribution $P$ over variables $X_1, \dots, X_n \in \{-1,1\}$ is an {\em Ising model} if for all $x \in \{-1,1\}^n$:
\begin{equation}\label{eqn:ising}
P(x) = \frac{\exp\left(\sum_{i\neq j \in [n]} A_{ij} x_i x_j + \theta \sum_{i \in [n]} x_i\right)}{{\sum_{z \in \set{-1,1}^n} \exp\left(\sum_{i\neq j \in [n]} A_{ij} z_i z_j + \theta \sum_{i\in[n]} z_i\right)}}
\end{equation}
where $\theta \in \R$ is called the {\em external field} and $A_{ij}$ are called the {\em interaction terms}. An Ising model is called {\em ferromagnetic} if all $A_{ij} \geq 0$. The {\em width} of an Ising model as in (\ref{eqn:ising}) is $\max_i \sum_j |A_{ij}| + |\theta|$. 

Invoking our framework on Ising models, we obtain:
\begin{restatable}{theorem}{isingmain}\label{thm:ising-main}
Let $\cD$ be the family of ferromagnetic Ising models having width at most $d$. Then, there is a distance approximation algorithm for $\cD$ with sample complexity $m = e^{O(d)}\epsilon^{-4}n^8\log ({n\over \epsilon})$ and runtime $O(mn^2+\epsilon^{-2} n^{17} \log n)$. 
\end{restatable}
We use the parameter learning algorithm by Klivans and Meka~\cite{KM17} that learns the parameters $\hat{\theta}, \hat{A}_{ij}$ of another Ising model $\hat{P}$ such that $\hat{P}(x)$ is a $(1\pm \eps)$ approximation of $P(x)$ for every $x$. This results holds for any Ising model, ferromagnetic or not. But in order to get an \eval approximator, we need to compute $\hat{P}(x)$ from $\hat{\theta}, \hat{A}_{ij}$. In general, the partition function (i.e., the sum in the denominator of \cref{eqn:ising}) may be $\mathsf{\#P}$-hard to compute, but for ferromagnetic Ising models, Jerrum and Sinclair~\cite{JS93} gave a PTAS for this problem. Thus, we obtain an $(\eps, \eps)$-\eval approximator for ferromagnetic Ising models that runs in polynomial time, and then \cref{thm:ising-main} follows from \cref{thm:maininf}.

Daskalakis et al.~\cite{DBLP:journals/tit/DaskalakisD019} studied independent testing and identity testing for Ising models and design {\em non-tolerent} testers.
Their sample and time complexity have polynomial dependence on the width instead of exponential (as in our case), but their algorithms seem to be inherently non-tolerant. In contrast, our distance approximation algorithm leads to a tolerant closeness-testing algorithm for ferromagnetic Ising models.  
Also, \cref{thm:ising-main} offers a template for distance approximation algorithms whenever the partition function can be approximated efficiently. In particular, Sinclair et al~\cite{SST14} showed a PTAS for computing the partition function of anti-ferromagnetic Ising models in certain parameter regimes.

We also show that we can efficiently approximate the distance to uniformity for any Ising model, whether ferromagnetic or not. Below, $U$ is the uniform distribution over $\{-1,1\}^n$.
\begin{restatable}{theorem}{isingsupp}\label{thm:ising-unif}
There is an algorithm which, given independent samples from an unknown Ising model  $P$ over $\{-1,1\}^n$ with width at most $d$, takes $m=O(e^{O(d)}\epsilon^{-4}n^8\log ({n/\epsilon})+\epsilon^{-7}\log^3 {1\over \epsilon})$ samples, $O(mn^2+\epsilon^{-7}n^2\log^3 {1\over \epsilon})$ time and returns a value $e$ such that $|e-\dtv(P,U)|\le \epsilon$ with probability at least 7/12, where $U$ is the uniform distribution over $\{-1,1\}^n$. 
\end{restatable}

The proof of \cref{thm:ising-unif} again proceeds by learning the parameters $\hat{\theta}, \hat{A}$ of an Ising model $\hat{P}$ that is a multiplicative approximation fo $P$. As we mentioned earlier, computing the partition function is in general hard, but now, we can efficiently estimate the ratio ${P}(x)/{P}(y)$ between any two $x,y \in \{-1,1\}^n$. At this point, we invoke the uniformity tester shown by Canonne et al.~\cite{DBLP:journals/siamcomp/CanonneRS15} that uses samples from the input distribution as well as pairwise conditional samples (the so-called \textsf{PCOND} oracle model). 

\subsection{Multivariate Gaussians}
\cref{thm:maininf} applies also when $\Sigma$ is not finite, e.g., the reals. Then, in the definition of the $(\beta,\gamma)$-\eval approximator $E_P$ for a distribution $P$, we require that there is a distribution $\hat{P}$ such that $\dtv(P,\hat{P})\leq \beta$ and $E_P$ is a $(1\pm \gamma)$-approximation of the {\em probability density function} of $\hat{P}$ at any $x$.

The most prominent instance in which we can apply our framework in this setting is for the class of multivariate gaussians, again another widely used model for high-dimensional distributions used throughout the natural and social sciences (see, e.g., \cite{MDLW18}). There are two main reasons for their ubiquity. Firstly, because of the central limit theorem, any physical quantity that is a population average is approximately distributed as a gaussian. Secondly, the gaussian distribution has maximum entropy among all real-valued distributions with a particular mean and covariance; therefore, a gaussian model places the least restrictions beyond the first and second moments of the distribution.

For $\mu \in \R^n$ and positive definite $\Sigma \in \R^{n \times n}$, the distribution $N(\mu, \Sigma)$ has the density function:
\begin{equation}\label{eqn:gaussian}
N(\mu, \Sigma; x) = \frac{1}{(2\pi)^{n/2} \sqrt{\det(\Sigma)}} \exp\left(-\frac12(x-\mu)^\top \Sigma^{-1} (x-\mu)\right)
\end{equation}
Invoking our framework on multivariate gaussians, we obtain:
\begin{restatable}{theorem}{gaussianmain}\label{thm:gaussians-main}
Let $\cD$ be the family of multivariate gaussian distributions, $\{ N(\mu, \Sigma): \mu \in \R^n, \Sigma \in \R^{n \times n}, \Sigma \succ 0\}$. Then, there is a distance approximation algorithm for $\cD$ with sample complexity $O(n^2 \eps^{-2})$ and runtime $O(n^\omega \eps^{-2})$ (where $\omega > 2$ is the matrix multiplication constant).
\end{restatable}
It is folklore that for any $P = N(\mu, \Sigma)$, the empirical mean $\hat{\mu}$ and empirical covariance $\hat{\Sigma}$ obtained from $O(n^2\eps^{-2})$ samples from $P$ determines a gaussian $\hat{P} = N(\hat{\mu}, \hat{\Sigma})$ satisfying $\dtv(P,\hat{P}) \leq \eps$ with probability at least 3/4. To get an \eval approximator, we need evaluations of $N(\hat{\mu}, \hat{\Sigma}; x)$ for any $x$ as in (\ref{eqn:gaussian}). Since $\det(\hat{\Sigma})$ is computable in time $O(n^\omega)$, \cref{thm:gaussians-main} follows from \cref{thm:maininf}.

This result is interesting because there is no closed-form expression known for the total variation distance between two gaussians of specified mean and covariance. Devroye et al.~\cite{DMR18} give expressions for lower- and upper-bounding the total variation distance that are a constant multiplicative factor away from each other. On the other hand, our approach (see \cref{cor:tv-approx}) yields a polynomial time randomized algorithm that, given $\mu_1, \Sigma_1, \mu_2, \Sigma_2$,  approximates the total variation distance $\dtv(N(\mu_1, \Sigma_1), N(\mu_2, \Sigma_2))$ upto $\pm \eps$ additive error.  

\subsection{Interventional Distributions in Causal Models}
A {\em causal model} for a system of random variables describes not only how the variables are correlated but also how they would change if they were to be externally set to prescribed values. To be more formal, we can use the language of  {\em causal Bayesian networks} due to Pearl~\cite{Pearl00}. A causal Bayesian network is a Bayesian network with an extra {\em modularity} assumption: for each node $i$ in the network, the dependence of $X_i$ on $X_{{\rm parents}(i)}$ is an autonomous mechanism that does not change even if other parts of the network are changed. 

Suppose $\cP$ is a causal Bayesian network over variables $X_1, \dots, X_n$ on a directed acyclic graph $G$ with nodes labeled $\{1,\dots, n\}$.
The nodes in $G$ are partitioned into two sets: {\em observable} $V$ and {\em hidden} $U$. A sample from the observational distribution $P$ yields the values of variables $X_V = \{X_i : \in V\}$. The modularity assumption allows us to define the result of {\em interventions} on causal Bayesian networks. An intervention is specified by a subset $S \subseteq V$ and an assignment $s \in \Sigma^{|S|}$. In the resulting interventional distribution, the variables in $S$ are fixed to $s$, while the variables $X_i$ for $i \notin S$ are sampled in topological order as it would have been in the original Bayesian network, according to the conditional probability distribution $X_i \mid X_{{\rm parents}(i)}$, where $X_{{\rm parents}(i)}$ consist of either variables previously sampled in the topological order or variables in $S$ set by the intervention. Finally, the variables in $U$ are marginalized out. The resulting distribution on $X_V$ is denoted $P_s$.

The question of inferring the interventional distribution from samples is a fundamental one. We focus on {\em atomic interventions}, i.e., where the intervention is on a single node $A \in V$. In this case, Tian and Pearl~\cite{TP02b, tian-thesis} exactly characterized the graphs $G$ such that for any causal Bayesian network $\cP$ on $G$ and for any assignment $a \in \Sigma$ to $X_A$, the interventional distribution $P_a$ is {\em identifiable}\footnote{That is, there exists a well-defined function mapping $P$ to $P_a$ but which may not be computationally effective.} from the observational distribution $P$ on $X_V$. For identification to be computationally effective, it is also natural to require a {\em strong positivity} condition on $P$, defined in \cref{sec:cbn}.
We show that we can efficiently estimate the distances between interventional distributions of causal Bayesian networks whenever the identifiability and strong positivity conditions are met:
\begin{theorem}[Informal]\label{thm:cbninf}
Suppose $\cP, \cQ$ are two unknown causal Bayesian networks on two known graphs $G_1$ and $G_2$ on a common observable set $V$ containing a special node $A$ and having bounded in-degree and c-component size. Suppose $G_1$ and $G_2$ both satisfy the identifiability condition, and the observational distributions $P$ and $Q$ satisfy the strong positivity condition. 

Then there is an algorithm which for any $a \in \Sigma$ and parameter $\eps \in (0,1)$ returns a value $e$ such that
$|e-\dtv(P_a,Q_a)|\le \epsilon$ with probability at least 2/3 using $\mathrm{poly}(|\Sigma|, n, \eps^{-1})$ samples from the observational distributions $P$ and $Q$ and running in time $\mathrm{poly}(|\Sigma|, n, \eps^{-1})$.
\end{theorem}

We again use the framework of \eval approximators to prove the theorem, but there is a complication: we do not get samples from the distributions $P_a$ and $Q_a$, but only from $P$ and $Q$. 
We build on a recent work (\cite{BGKMV20}) that shows how to efficiently learn and sample from interventional distributions of atomic interventions using observational samples, assuming the identifiability and strong positivity conditions. 

\cref{thm:cbninf} solves a very natural problem. To concoct a somewhat realistic example, suppose a biologist wants to compare how a particular point mutation affects the activity of other genes for Africans and for Europeans. Because of ethical reasons, she cannot conduct randomized controlled trials by actively inducing the mutation, but she can draw random samples from the two populations. It is reasonable to assume that the graph structure of the regulatory network is the same for all individuals, and we further assume that the causal graph over the genes of interest is known (or can be learned through other methods). Also, suppose that the gene expression levels can be discretized. She can then, in principle, use the algorithm proposed in \cref{thm:cbninf} to test whether the effect of the mutation is approximately the same for Africans and Europeans.

\subsection{Improving Success of Learning Algorithms Using Distance Estimation}
Finally we give a link between efficient distance approximation algorithms and boosting the success probability of learning algorithms.  Specifically, let $\mathcal{D}$ be a family of distributions for which we have a learning algorithm $\mathcal{A}$ in $\dtv$ distance $\epsilon$ that succeeds with probability 3/4. Suppose there is also a distance approximation algorithm $\mathcal{B}$ for $\mathcal{D}$. We prescribe a method to  combine the two algorithms $\mathcal{A}$ and $\mathcal{B}$ to learn an unknown distribution from $\mathcal{D}$ with probability at least $(1-\delta)$. To the best of our knowledge, this connection has not been stated explicitly in the literature. The proof of the following theorem is given in \cref{sec:boost}.

\begin{restatable}{theorem}{boostmain}\label{thm:boost}
Let $\mathcal{D}$ be a family of distributions. Suppose there is a learning algorithm $\mathcal{A}$ which for any $P\in \mathcal{D}$ takes $m_{\mathcal{A}}(\epsilon)$ samples from $P$ 
and in time $t_{\mathcal{A}}(\epsilon)$ outputs a distribution $P_1$ such that $\dtv(P,P_1)\le \epsilon$ with probability at least 3/4. Suppose there is a distance approximation 
algorithm $\mathcal{B}$ for $\mathcal{D}$ that given any two completely specified distributions $P_1$ and $P_2$ estimates $\dtv(P_1,P_2)$ up to an additive error $\epsilon$ in 
$t_{\mathcal{B}}(\epsilon,\delta)$ time with probability at least $(1-\delta)$. Then there is an algorithm that uses $\mathcal{A}$ and $\mathcal{B}$ as subroutines, takes $O(m_{\mathcal{A}}
(\epsilon/4)\log {1\over \delta})$ samples from $P$, runs in $O(t_{\mathcal{A}}(\epsilon/4)\log {1\over \delta}+t_{\mathcal{B}}(\epsilon/4,{\delta\over 210000 \log^2 {2\over \delta}})\log^2 {1\over \delta})$ time and returns a distribution $\hat{P}$ such that $\dtv(P,\hat{P})\le \epsilon$ with probability at least $1-\delta$.  
\end{restatable}

To achieve the above result we repeat $\mathcal{A}$ independently $R=O(\log {1\over \delta})$ times which guarantees at least $2R/3$ successful repetitions from Chernoff's bound except $\delta$ probability, which we condition on. Sucessful repetitions must produce distributions which are pairwise $2\epsilon$ close by triangle inequality. We approximate the pairwise distances between all pairs of repetitions up to an additive $\epsilon$ and then find out a repetition whose learnt distribution $\hat{P}$ has the most number of other repetitions within $3\epsilon$ distance. The later number must be at least $2R/3-1$, guaranteeing $\hat{P}$ must have a successful repetition within $3\epsilon$ distance. Thus $\hat{P}$ must be at most $4\epsilon$ close to $P$ from triangle inequality.

\ignore{\subsection{Organization}
So far we have summarized our main findings regarding distance approximation for structured high dimensional distributions. In Section 
2 we present our main algorithm which approximates distance between two distributions which can be evaluated as well as sampled. In Section 3 we apply our main algorithm to approximate the distance between two  high dimensional bayesian networks. Due to paucity of space we could not provide details for all other findings which could be found in the Appendix.
}

\subsection{Previous work}\label{sec:prior}

Prior work most related to our work is in the area of distribution
testing. The topic of distribution testing is rooted in statistical
hypothesis testing and goes back to Pearson's chi-squared test in
1900. In theoretical computers science, distribution testing research
is relatively new and focuses on designing hypothesis testers with
optimal sample complexity. Goldreich and Ron~\cite{GoldreichR11}
investigated uniformity testing (distinguishing whether an input
distribution $P$ is uniform over its support or $\eps$-far from
uniform in total variation distance) and designed a tester with sample
complexity $O(m/\eps^4)$ (where $m$ is the size of the sample
space). Paninski~\cite{Pan08} showed that $\Theta(\sqrt{m}/\eps^2)$ samples are
necessary for uniformity testing, and gave an optimal tester when
$\epsilon>m^{-1/4}$.  Batu et al.~\cite{BatuFRSW13} initiated the
investigation of identity (goodness-of-fit) testing and closeness
(two-sample) testing and gave testers with sample complexity
$\tilde{O}(\sqrt{m}/\eps^6)$ and $\tilde{O}(m^{2/3}
\textrm{poly}(1/\eps))$ respectively. Optimal bounds
for these testing problems were obtained in
Valiant and Valiant~\cite{Valiant:2014:AIP:2706700.2707449} ($\Theta(\sqrt{m}/\eps^2)$)
and Chan et al.~\cite{ChanDVV14}
($\Theta(\max(m^{2/3}\eps^{-4/3},\sqrt{m}\eps^{-2}))$)
respectively. Tolerant versions of these testing problems have very
different sample complexity.  In particular, Valiant and Valiant~\cite{ValiantV11,ValiantV10} showed that tolerant uniformity,
identity, and closeness testing with respect to the total variation
distance have a sample complexity of $\Theta(m/\log
m)$. Since the seminal papers of
Goldreich and Ron and Batu et al., distribution testing grew into a
very active research topic and a wide range of properties of
distributions have been studied under this paradigm. This research led
to sample-optimal testers for many distribution properties.  We refer
the reader to the surveys \cite{DBLP:journals/eccc/Canonne15,DBLP:journals/crossroads/Rubinfeld12} and references
therein for more details and results on the topic.

When the sample space is a
high-dimensional space (such as $\{0,1\}^n)$), the testers designed
for general distributions require exponential number of samples
($2^{\Omega(n)}$) if the sample space is $\{0,1\}^n$ for a constant
$\epsilon$). Thus structural assumptions are to be made to design
efficient ($\poly(n,1/\epsilon)$) and practical testers for many of
the testing problems. The study of testing high-dimensional
distributions with structural restrictions was initiated only very
recently. The work that is most closely related to our work appears
in~\cite{DBLP:journals/tit/DaskalakisD019,DBLP:conf/colt/CanonneDKS17,DBLP:conf/colt/DaskalakisP17,10.5555/3327546.3327616}
(these works also give good expositions to other prior work on this
topic). These papers consider distributions coming from graphical
models including Ising models and Bayes nets.
In Daskalakis et al.~\cite{DBLP:journals/tit/DaskalakisD019}, the authors consider distributions that are
drawn from an Ising model and show that identity testing and {\em
  independence testing} (testing whether an unknown distribution is
close to a product distribution) can be done with
$\poly(n,1/\epsilon)$ samples where $n$ is the number nodes in the
graph associated with the Ising model.
In Canonne et al.~\cite{DBLP:conf/colt/CanonneDKS17}
and Daskalakis et al.~\cite{DBLP:conf/colt/DaskalakisP17}, the authors consider identity
testing and closeness testing for distributions given by Bayes
networks of bounded in-degree. Specifically, they design algorithms
with sample complexity $\tilde{O}(2^{3(d+1)/4} n/\eps^2)$ that test
closeness of distributions over the same Bayes net with $n$ nodes and
in-degree $d$. They also show that $\Theta(\sqrt{n}/\eps^2)$ and
$\Theta(\max(\sqrt{n}/\eps^2, n^{3/4}/\eps))$ samples are necessary
and sufficient for identity testing and closeness testing respectively
of pairs of product distributions (Bayes net with empty
graph). Finally, in Acharya et al.\cite{10.5555/3327546.3327616}, the
authors investigate testing problems on {\em causal Bayesian networks}
as defined by Pearl~\cite{Pearl00} and design efficient
$(\poly(n,1/\eps))$ testing algorithms for certain identity and
closeness testing problems for them. All these papers consider
designing non-tolerant testers and leave open the problem of designing
efficient testers that are tolerant for high-dimensional distributions
which is the main focus in this paper.

Our main technical result builds on the work of Canonne and Rubinfeld~\cite{DBLP:conf/icalp/CanonneR14}.  They consider a {\em
  dual access model} for testing distributions. In this model, in
addition to independent samples, the testing algorithm has also access
to an evaluation oracle that gives probability of any item in the
sample space. They establish that having access to evaluation oracle
leads to testing algorithms with sample complexity independent of the
size of the sample space. Indeed, in order to design testing
algorithms, they give an algorithm to additively estimate the total
variation distance between two unknown distributions in the dual
access model.  Our distance estimation algorithm is a direct extension
of this algorithm.

Another access model considered in the literature
for which such domain independent results are obtained is the {\em
  conditional sampling model} introduced independently in
Chakraborty et al.~\cite{DBLP:journals/siamcomp/ChakrabortyFGM16} and Canonne et al.~\cite{DBLP:conf/soda/CanonneRS14}. In this model, the tester has access to a conditional sampling oracle that given a
subset $S$ of the sample space outputs a random sample from the
unknown distribution {\em conditioned on $S$}. The conditional sampling model lends itself to algorithms for testing uniformity and testing identity to a known distribution with sample complexity  $\tilde{O}(1/\eps^2)$. Building on Chakraborty et al.~\cite{DBLP:journals/siamcomp/ChakrabortyFGM16}, Chakraborty and Meel~\cite{CM19} proposed a tolerant testing algorithm with sample complexity independent of domain size for testing uniformity of a sampler that takes in a Boolean formula $\varphi$ as input and the sampler's output generates a distribution over the witnesses of $\varphi$. \ignore{We note that some of the results obtained in this work for dual access model can be extended to conditional sampling model as well.}


\section{Distance Approximation Algorithm}\label{sec:approx}

In this section, we prove \cref{thm:maininf} which underlies all the other results in this work. In fact, we show
the following theorem that is more detailed.

\begin{theorem}\label{thm:main}
Suppose we have sample access to distributions $P$ and $Q$ over a finite set. Also, suppose we can make calls to two circuits $\cC_P$ and $\cC_Q$ which implement $(\beta,\gamma)$-\eval approximators for $P$ and $Q$ respectively.   Let $T$ be the maximum running time for any call to $\cC_P$ or $\cC_Q$. 

Then for any $\epsilon, \delta>0$, $\dtv(P,Q)$ can be approximated up to an additive error ${2\gamma\over 
1-\gamma} + 3\beta+ \epsilon$ with probability at least $1-\delta$, using $O(\epsilon^{-2}\log \delta^{-1})$ samples from $P$ and $O(\epsilon^{-2}\log \delta^{-1} \cdot T)$ runtime.
\end{theorem}

Note that the \eval approximators  in \cref{thm:main} must return rational numbers with bounded denominators as they are implemented by circuits with bounded running time. The exact model of computation for the circuits does not matter so much, so we omit its discussion. 

We now turn to the proof of \cref{thm:main}. As mentioned in the Introduction, if $\cC_P$ and $\cC_Q$ were $(0,0)$-\eval approximators, the result already appears in \cite{DBLP:conf/icalp/CanonneR14}. The proof below analyzes how having nonzero $\beta$ and $\gamma$ affects the error bound.
\ignore{
We will consider high dimensional distributions over $\Sigma^n$ for some alphabet $\Sigma$. Our goal is to design algorithms that is efficient in both time and sample complexity that  has polynomial dependency on $n$. In this section, we do not explicitly write the efficiency parameters as the exact dependency will be determined by the class of probability distributions we are interested in. 

For a distribution $P$ and an element $i$ in the sample space, $P_i$ denotes $i$'s probability.  

\begin{definition}\label{def:oracle}
Let $P$ be a distribution over $\Sigma^n$. A function $E_P:\Sigma^n \rightarrow [0,1]$ is a $(\beta,\gamma)$ eval approximator for $P$ if there exists a distribution $\hat{P}$ over $\Sigma^n$ such that
\begin{itemize}
\item $\dtv(P,\hat{P})\le \beta$
\item $\forall i \in \Sigma^n$, $(1-\gamma)\hat{P}_i \le E_P(i) \le (1+\gamma)\hat{P}_i$
\end{itemize}
\end{definition}

For two distributions having eval approximator as above, we give Algorithm~\ref{algo:main} for estimating their total variation distance efficiently.

\begin{theorem}\label{thm:main}
Let $P$ and $Q$ be two distributions for which $E_P$ and $E_Q$ are the two $(\beta,\gamma)$ eval approximators respectively.  Then for any $\epsilon>0$, $\dtv(P,Q)$ can be approximated up to an additive error ${2\gamma\over 
1-\gamma} + 3\beta+ \epsilon$ with probability $\geq (1-\delta)$, using $m$ samples from $P$ and $m$ calls to $E_P$ and $E_Q$ where $m=O(\log {1\over \delta}/\epsilon^2)$.
\end{theorem}
}

\begin{figure}
\begin{algorithm}[H]
\SetKwInOut{Input}{Input}
\SetKwInOut{Output}{Output}
\Input{Sample access to distribution $P$; oracle access to circuits $\cC_P$ and $\cC_Q$.}
\Output{Approximate value of $\dtv(P,Q)$}
\For{$i=1, \dots, t=O(\eps^{-2}\log {\delta^{-1}})$}{
Draw a sample $x$ from $P$\;
$\alpha \gets \cC_P(x)$\;
$\beta \gets \cC_Q(x)$\;
$c_i \gets 1_{\alpha>\beta}\left(1-\frac{\beta}{\alpha}\right)$;
}
\Return{${1\over t}\sum_{i=1}^t c_i$}
\caption{Distance approximation}
\label{algo:main}
\end{algorithm}
\end{figure}
\begin{proof}
We invoke \cref{algo:main}. Notice that the algorithm only requires sample access to one of the two distributions but to both of the \eval approximators. Let $\hat{P}$ be the distribution $\beta$-close to $P$ which is approximated by the output of $\cC_P$; similarly define $\hat{Q}$. 

We have $|\dtv(P,Q)-\dtv(\hat{P},\hat{Q})|\le \dtv(P,\hat{P})+\dtv(Q,\hat{Q}) \le 2\beta$ from the triangle inequality. Hence, it is sufficient to approximate $\dtv(\hat{P},\hat{Q})$ additively up to  ${2\gamma\over 
1-\gamma} + \beta+ \epsilon$.  
{\allowdisplaybreaks
\begin{align*}
\dtv(\hat{P},\hat{Q})&={1\over 2}\sum_x |\hat{P}(x)-\hat{Q}(x)|\\
&=\sum_{x:\hat{P}(x)>\hat{Q}(x)} (\hat{P}(x)-\hat{Q}(x))\\
&=\sum_{x:\hat{P}(x)>\hat{Q}(x)} \left(1-\frac{\hat{Q}(x)}{\hat{P}(x)}\right)\hat{P}(x) &&\tag{Since $\hat{P}(x)> 0$}\\
&=\E_{x \sim \hat{P}}\left[1_{\hat{P}(x)>\hat{Q}(x)}\left(1-\frac{\hat{Q}(x)}{\hat{P}(x)}\right)\right]\\
\end{align*}}

From the above, if we have complete access (both evaluation and sample) to $\hat{P}$ and $\hat{Q}$, then we can estimate the distance with $O({1\over \epsilon^2}\log {1\over \delta})$ samples and evaluations. However as we have only approximate evaluations of $\hat{P}$ and $\hat{Q}$ and samples from the original distribution $P$, we need some additional arguments. Let $E_P$ and $E_Q$ be the functions implemented by the circuits $\cC_P$ and $\cC_Q$ respectively.

\begin{align*}
\dtv(\hat{P},\hat{Q})&=\sum_{x} 1_{\hat{P}(x)>\hat{Q}(x)}\left(1-\frac{\hat{Q}(x)}{\hat{P}(x)}\right)\hat{P}(x)\\
&= \underbrace{\sum_{x} 1_{E_P(x)>E_Q(x)}\left(1-\frac{E_Q(x)}{E_P(x)}\right)\hat{P}(x)}_A+ \\ 
& \qquad\underbrace{\sum_{x} \left[1_{\hat{P}(x)>\hat{Q}(x)}\left(1-\frac{\hat{Q}(x)}{\hat{P}(x)}\right)-1_{E_P(x)>E_Q(x)}\left(1-\frac{E_Q(x)}{E_P(x)}\right)\right]\hat{P}(x)}_B\\
\end{align*}

We start with an upper bound for the absolute value of the error term $B$. We consider the partition of sample space into $S_1,S_2$ and $S_3$, where $S_1=\{x:1_{\hat{P}(x)>\hat{Q}(x)}=1_{E_P(x)>E_Q(x)}\}$, $S_2=\{x:1_{\hat{P}(x)>\hat{Q}(x)}>1_{E_P(x)>E_Q(x)}\}$ and $S_3=\{x:1_{\hat{P}(x)>\hat{Q}(x)}<1_{E_P(x)>E_Q(x)}\}$.
{\allowdisplaybreaks
\begin{align*}
|B|&=\left|\sum_{x} \left[1_{\hat{P}(x)>\hat{Q}(x)}\left(1-\frac{\hat{Q}(x)}{\hat{P}(x)}\right)-1_{E_P(x)>E_Q(x)}\left(1-\frac{E_Q(x)}{E_P(x)}\right)\right]\hat{P}(x)\right|\\
&\le \sum_{x} \left|\left[1_{\hat{P}(x)>\hat{Q}(x)}\left(1-\frac{\hat{Q}(x)}{\hat{P}(x)}\right)-1_{E_P(x)>E_Q(x)}\left(1-\frac{E_Q(x)}{E_P(x)}\right)\right]\hat{P}(x)\right|\\
&= \sum_{x \in S_1} 1_{\hat{P}(x)>\hat{Q}(x)} \left|\frac{\hat{Q}(x)}{\hat{P}(x)}-\frac{E_Q(x)}{E_P(x)}\right| \hat{P}(x)+ \sum_{x \in S_2} 1_{\hat{P}(x)>\hat{Q}(x)} \left(1-\frac{\hat{Q}(x)}{\hat{P}(x)}\right)\hat{P}(x)  +\\
&\qquad\qquad\sum_{x \in S_3} 1_{E_P(x)>E_Q(x)} \left(1-\frac{E_Q(x)}{E_P(x)}\right)\hat{P}(x) \\
\end{align*}}
For $x$ in $S_1$ with $\hat{P}(x)>\hat{Q}(x)$, ${(1-\gamma)\over(1+\gamma)}{\hat{Q}(x)\over \hat{P}(x)} \le {E_Q(x)\over E_P(x)} \le {(1+\gamma)\over(1-\gamma)}{\hat{Q}(x)\over \hat{P}(x)}$ so that
$
\left|{\hat{Q}(x)\over \hat{P}(x)} - {E_Q(x)\over E_P(x)}\right| \le {2\gamma\over 1-\gamma} {\hat{Q}(x)\over \hat{P}(x)}<{2\gamma\over 1-\gamma}.$
For $x$ in $S_2$, $\hat{P}(x)>\hat{Q}(x)$ implies $E_P(x)\leq E_Q(x)$ and hence, $(1-\gamma) \hat{P}(x) \le E_P(x) \le E_Q(x) \le (1+\gamma)\hat{Q}(x)$ so that $\hat{Q}(x)/\hat{P}(x)\ge \frac{1-\gamma}{1+\gamma}$. For $x$ in $S_3$, $E_P(x)>E_Q(x)$ implies $\hat{P}(x) \leq \hat{Q}(x)$, and hence, $\frac{E_Q(x)}{E_P(x)} \ge \frac{(1-\gamma)\hat{Q}(x)}{(1+\gamma)\hat{P}(x)} \ge \frac{1-\gamma}{1+\gamma}$. Therefore:
\begin{align*}
|B| &\le \sum_{x\in S_1}{2\gamma\over 1-\gamma} \hat{P}(x) + \sum_{x\in S_2} {2\gamma\over 1+\gamma} \hat{P}(x) + \sum_{x\in S_3} {2\gamma\over 1+\gamma} \hat{P}(x)\\ 
&\le {2\gamma\over 1-\gamma}
\end{align*}

Now consider the term $A$:
\begin{align*}
A&=\sum_{x} 1_{E_P(x)>E_Q(x)}\left(1-\frac{E_Q(x)}{E_P(x)}\right)\hat{P}(x)\\
&= \underbrace{\sum_{x} 1_{E_P(x)>E_Q(x)}\left(1-\frac{E_Q(x)}{E_P(x)}\right){P}(x)}_C + \sum_{x} 1_{E_P(x)>E_Q(x)}\left(1-\frac{E_Q(x)}{E_P(x)}\right)(\hat{P}(x)-P(x)).
\end{align*}
Note that: $\left|\sum_{x} 1_{E_P(x)>E_Q(x)}\left(1-\frac{E_Q(x)}{E_P(x)}\right)(\hat{P}(x)-P(x))\right|\leq \sum_x |\hat{P}(x)-P(x)| \leq \beta$. So, $|\dtv(\hat{P},\hat{Q}) - C | \le {2\gamma\over 1-\gamma}+\beta$. 
We can rewrite $C$ as $\argEx{x\sim P} {1_{E_P(x)>E_Q(x)}\left(1-{ E_Q(x) \over E_P(x)}\right)}$.
Since $1_{E_P(x)>E_Q(x)}\left(1-{ E_Q(x) \over E_P(x)}\right)$ lies in $[0,1]$, by the Chernoff bound, we can estimate the expectation up to $\epsilon$ additive error with probability at least $(1-\delta)$ by averaging $O({1\over \epsilon^2}\log {1\over \delta})$ samples from $P$. 
\end{proof}

\cref{thm:main} can be extended to the case that $P$ and $Q$ are distributions over $\R^n$ with infinite support. We change \cref{def:oracle} so that $E_P(x)$ is a $(1\pm \gamma)$-approximation of  $\hat{f}(x)$ where $\hat{f}(x)$ is the probability density function for $\hat{P}$. Then, \cref{thm:main} and \cref{algo:main} continue to hold as stated. In the proof, we merely have to replace the summations with the appropriate integrals. 


\section{Bayesian networks}\label{sec:bn}
First we apply our distance estimation algorithm for tolerant testing of high dimensional distributions coming from bounded in-degree Bayesian networks. Bayesian networks defined below are popular probabilistic graphical models for describing high-dimensional distributions succinctly. 

\begin{definition}
A  {\em Bayesian network} $P$ on a directed acyclic graph $G$ over the vertex set $[n]$ is a joint distribution of the $n$ random variables $(X_1,X_2,\dots,X_n)$ over the sample space $\Sigma^n$ such that for every $i \in [n]$ $X_i$ is conditionally independent of $X_{\emph{non-descendants}(i)}$ given $X_{\emph{parents}(i)}$, where for $S\subseteq [n]$, $X_S$ is the joint distribution of $(X_i: i \in S)$, and parents and non-descendants are defined from $G$. 

$P$ factorizes as follows:
\begin{equation}\label{eqn:bnsecbn}
P(x) \coloneqq \Pr_{X \sim P}[X=x]= \prod_{i=1}^n \Pr_{X\sim P}[X_i = x_i \mid \forall j \in {\rm parents}(i), X_j = x_j] \qquad \text{for all } x \in \Sigma^n
\end{equation}

Hence a Bayesian network can be completely described by a set of conditional distributions for every variable $X_i$, for every fixing of its parents $X_{\emph{parents}(i)}$. 
\end{definition}
\ignore{
For a Bayesian network distribution $P$, we denote the underlying graph by $G_P$. 
Daskalakis and Pan~\cite{DBLP:conf/colt/DaskalakisP17} gave a sample- and time-efficient algorithm for testing $P=Q$ versus $\dtv(P,Q)\ge \epsilon$ where $P$ and $Q$ are two unknown distributions coming from a common Bayesian network with a small indegree $(\le d)$ on $n$ variables with alphabet $\Sigma$ whose underlying graph is known. Their algorithm takes $m=\tilde{O}(|\Sigma|^{3(d+1)/4}n/\epsilon^2)$ samples and $O(mn)$ runtime and succeeds with probability 2/3. When the graph is common but unknown but its maximum indegree $d$ is known they gave an alternate algorithm with the same sample complexity and $O(mn^d)$ runtime. In this section we give a sample- and time-efficient algorithm that takes two Bayesian networks on potentially different but known pairs of graphs and distinguishes the cases $\dtv(P,Q)< \epsilon/2$ versus $\dtv(P,Q)\ge \epsilon$. Thus our tester has tolerance and it removes the assumption that the two underlying graphs need to be the same. More generally our tester is able to handle the case when $Q$ is any distribution with an efficient oracle as defined in Definition~\ref{def:oracle}.
}

To construct an \eval approximator for a Bayesian network, we first learn it using an efficient algorithm. Such a learning algorithm was claimed in the appendix of \cite{DBLP:conf/colt/CanonneDKS17} but the analysis there appears to be incomplete \cite{CanComm}. We 
show the following proper learning algorithm for Bayesian networks that uses the optimal sample complexity.
\ignore{This approach is known as testing-by-learning approach in distribution testing. For unstructured distributions on sample space of size $n$ such an approach becomes an overkill since identity testing for them can be performed with $O(\sqrt{n})$~\cite{Valiant:2014:AIP:2706700.2707449} samples while learning would require $\Omega(n)$ samples. For Bayesian networks identity testing has a lower bound which scales as $\Omega(n)$~\cite{DBLP:journals/tit/DaskalakisD019} and the sample complexity of a learning algorithm for them scales as $O(n)$~\cite{DBLP:conf/colt/CanonneDKS17}.}

\begin{theorem}\label{thm:bnlearning}
There is an algorithm that given a parameter $\epsilon>0$ and sample access to an unknown Bayesian network distribution $P$ on a known directed acyclic graph $G$ of in-degree at most $d$, returns a Bayesian network $\hat{P}$ on $G$ such that $\dtv(P,\hat{P}) \leq \epsilon$ with probability $\geq 9/10$. Letting $\Sigma$ denote the range of each variable $X_i$, the algorithm takes $m=O(|\Sigma|^{d+1}n\log(|\Sigma|^{d+1} n)\epsilon^{-2})$ samples and runs in $O(|\Sigma|^{d+1} mn)$ time.  
\end{theorem}
\ignore{
\begin{theorem}[\cite{DBLP:journals/tit/DaskalakisD019}]\label{thm:bntestinglb}
Given sample access to two unknown Bayesian network distributions $P$ and $Q$ over $\{0,1\}^n$ so that $G_P=G_Q$ and $G_P$ is known,  testing $P=Q$ versus $\dtv(P,Q)\ge \epsilon$ with probability $\geq 2/3$ requires $\Omega(n/\epsilon^2)$ samples. 
\end{theorem}
}
\ignore{ Our algorithm has optimal dependence on $n$ and $\epsilon$ from Theorem~\ref{thm:bntestinglb}
This directly gives us a distance estimation algorithm for Bayesian networks. 
\begin{theorem}
Suppose $G_1$ and $G_2$ are two known DAGs on $n$ vertices with in-degree at most $d$. Let $P_1$ and $P_2$ be two unknown Bayesian networks over $\Sigma^n$ on $G_1$ and $G_2$ respectively.
Then, there is an algorithm  that takes $m=\tilde{O}(|\Sigma|^{d+1} n \eps^{-2})$ samples, $O(|\Sigma|^{d+1} mn)$ time and returns a number $e$ such that $|e-\dtv(P_1,P_2)|\le \epsilon$ with probability 2/3.
\end{theorem}
}
This directly gives us a distance estimation algorithm for Bayesian networks. 

\bnmain*
\begin{proof}
Given samples from $P_1$ and $P_2$ we first learn them as $\hat{P}_1$ and $\hat{P}_2$ using \cref{thm:bnlearning} in $\dtv$ distance $\epsilon/4$. This step costs $m=O(|\Sigma|^{d+1}n\log(|\Sigma|^{d+1} n)\epsilon^{-2})$ samples and $O(|\Sigma|^{d+1}mn)$ time and succeeds with probability 4/5. $\hat{P}_1$ and $\hat{P}_2$ gives efficient $(\epsilon/4,0)$-\eval approximators from~\cref{eqn:bnsecbn}. It follows from~\cref{thm:main} that we can estimate $\dtv(P_1,P_2)$ up to an $\epsilon$ additive error using $O(\epsilon^{-2})$ additional samples from $P_1$ except for 1/5 probability.
\end{proof}

Our distance estimation algorithm has optimal dependence on $n$ and $\epsilon$ from the following non-tolerant identity testing lower bound of Daskalakis et al.
\begin{theorem}[\cite{DBLP:journals/tit/DaskalakisD019}]\label{thm:bntestinglb}
Given sample access to two unknown Bayesian network distributions $P_1$ and $P_2$ over $\{0,1\}^n$ on a common known graph,  testing $P=Q$ versus $\dtv(P,Q)\ge \epsilon$ with probability $\geq 2/3$ requires $\Omega(n\epsilon^{-2})$ samples. 
\end{theorem}
It remains to prove \cref{thm:bnlearning}.

\subsection{Learning Bayesian networks}\label{sec:app-bn}

\ignore{We design a learning algorithm for Bayesian networks on a known DAG $G$ that uses $\tilde{O}(n \epsilon^{-2} |\Sigma|^{d+1})$ samples where $d$ is the maximum in-degree. We note 
that in \cite{DBLP:conf/colt/CanonneDKS17}, the authors give a sketch of a learning algorithm that claimed to get the same sample complexity bound. However their analysis is incomplete~\cite{CanonnePersonal}.}

In this section, we prove a strengthened version of \cref{thm:bnlearning} that holds for any desired error probability $\delta$.
\begin{theorem}\label{thm:constLearnBN}
There is an algorithm that given parameters $\epsilon, \delta>0$ and sample access to an unknown Bayesian network distribution $P$ on a known directed acyclic graph $G$ of in-degree at most $d$, returns a Bayesian network $Q$ on $G$ such that $\dtv(P,Q) \leq \epsilon$ with probability $\geq (1-\delta)$. Letting $\Sigma$ denote the alphabet for each variable $X_i$, the algorithm takes $m=O(|\Sigma|^{d+1}n\log(|\Sigma|^{d+1} n)\epsilon^{-2}\log {1\over \delta})$ samples and runs in $O(|\Sigma|^{d+1} mn)$ time.  
\end{theorem}

We actually prove a stronger bound on the distance between $P$ and $Q$ in terms of the $\KL$ divergence. The $\KL$ divergence between two distributions $P$ and $Q$ is defined as ${\rm KL}(P,Q)=\sum_i P(i) \ln {P(i) \over Q(i)}$. From Pinsker's inequality, we have $\dtv^2(P,Q) \le 2 \KL(P,Q)$. Thus a  $\dtv$ learning result follows from a $\KL$ learning result. We present \cref{algo:bn} for the binary alphabet case ($\Sigma = \{0,1\}$) and reduce the general case to the binary case afterwards. 

The {\em Laplace corrected empirical estimator} takes $z$ samples from a distribution over $k$ items and assigns to item $i$ the probability $(z_i+1)/(z+k)$ where $z_i$ is the number of occurrences of item $i$ in the samples. 
We will use the following general result for learning a distribution in $\KL$ distance.

\begin{theorem}[\cite{pmlr-v40-Kamath15}]\label{thm:exKLlearn}
Let $D$ be an unknown distribution over $k$ items.  Let $\hat{D}$ be the Laplace corrected empirical distribution of $z$ samples from $D$. Then for $k\ge 2, z\ge 1$, $\ex{\KL(D,\hat{D})} \le (k-1)/(z+1)$.
\end{theorem}

We will use a $\KL$ local additivity result for Bayesian networks, a proof of which is given in  ~\cite{DBLP:conf/colt/CanonneDKS17}. For a Bayesian network $P$, a  vertex $i$, and a  setting a value  $a$ of its parents, let $\Pi[i,a]$ denote the event that parents of $i$ take value $a$, and let $P(i \mid a)$ denote the distribution at vertex $i$ when its parents takes value $a$. 

\begin{theorem}\label{thm:klsub}
 Let $P$ and $Q$ be two Bayesian networks over the same graph $G$. Then  
$$\KL(P,Q) = \sum_{i} \sum_a P[\Pi[i,a]]\cdot \KL(P(i\mid a),Q(i\mid a))$$
\end{theorem}

\begin{figure}
\begin{algorithm}[H]
\SetKwInOut{Input}{Input}
\SetKwInOut{Output}{Output}
\Input{Samples from an unknown Bayesian network $P$ over $\{0,1\}^n$ on a known graph $G$ of in-degree $\le d$, parameters $m,t$}
\Output{A Bayesian network $Q$ over $G$}
Get $m$ samples from $P$\;
\For{every vertex $i$}{
\For{every fixing $a$ of $i$'s parents}{
$N_{i,a} \gets$ the number of samples where $i$'s parents are set to $a$\;
\eIf{$N_{i,a} \ge t$}{
$Q(i\mid a) \gets$ the Laplace corrected empirical distribution at node $i$ in the subset of samples where $i$'s parents are set to $a$\;
}{
$Q(i\mid a) \gets$ uniformly random bit\;
} 
}
}
\caption{Fixed-structure Bayesian network learning}
\label{algo:bn}
\end{algorithm}
\end{figure}

\begin{lemma}\label{thm:bnlearningapp}
For $m= 24n2^d \log (n2^d)/\epsilon$ and $t=12\log (n2^d)$, \cref{algo:bn} satisfies $KL(P,Q)\le 6\epsilon$ with probability at least 3/4 over the randomness of sampling.
\end{lemma}
\begin{proof}
Call a tuple $(i,a)$ {\em heavy} if $P[\Pi[i,a]]\ge \frac{\epsilon}{2^d n}$ and {\em light} otherwise. 
Let $N_{i,a}$ denote the number of samples where $i$'s parents are $a$. 

Consider the event  ``all heavy $(i,a)$ tuples satisfy $N_{i,a}\ge n2^dP[\Pi[i,a]]t/\epsilon$". It is easy to see from Chernoff and union bounds that this event holds with 19/20 probability.  Hence for the rest of the argument, we condition on this event. In this case, all heavy items satisfy  $N_{i,a}\ge t$.

Now, we see that:
\begin{itemize} 
\item For any heavy $(i,a)$, by \cref{thm:exKLlearn}, \[\ex{\KL(P(i\mid a),Q(i\mid a))} \le \frac{\epsilon}{10n2^d\cdot P[\Pi[i,a]]}.\]
\item
 For any light $(i,a)$ that satisfies $N_{i,a}\ge t$, it follows from \cref{thm:exKLlearn} that $\ex{\KL(P(i\mid a),Q(i\mid a))} \le 1$. 
 \item 
 Items which do not satisfy $N_{i,a}\ge t$ must be light for which $\KL(P(i\mid a),Q(i \mid a)) \le p\ln 2p + (1-p)\ln 2(1-p) \le 1$ where $p=P[i=1|a]$, since in that case $Q(i \mid a)$ is the uniform bit. 
 \end{itemize}
 
Using \cref{thm:klsub}, we get \[\ex{\KL(P,Q)}\le \sum_{(i,a) \text{ heavy}} P[\Pi[i,a]] \cdot \frac{\epsilon}{10n2^d \cdot P[\Pi[i,a]]} + \sum_{(i,a) \text{ light}} \frac{\epsilon}{n2^d}\cdot 1 \le 1.1\epsilon.\] The lemma follows from Markov's inequality.
\end{proof}

Now we reduce the case when $\Sigma$ is not binary to the binary case. We can encode each $\sigma\in \Sigma$ of the Bayesian network as a $\log |\Sigma|$ size boolean string which gives us a Bayesian network of degree $(d+1) \log |\Sigma|$ over $n\log |\Sigma|$ variables. Then we apply \cref{thm:bnlearningapp} to get a learning algorithm with $O(\epsilon)$ error in $\dtv$ and 3/4 success probability. Subsequently we repeat $O(\log {1\over \delta})$ times and find out a successful repetition using \cref{thm:boost}.

\section{Ising Models}\label{sec:Ising}
In this section, we give a distance approximation algorithm for the class of bounded-width ferromagnetic Ising models. Recall from \cref{sec:intro-ising} that a probability distribution $P$ from this class is over the sample space $\{-1,1\}^n$ and that $P(x)$, the probability of an item $x\in \{-1,1\}^n$, is proportional to the numerator: 
\[N(x)=\exp\left(\sum_{i,j} A_{i,j}x_i x_j +\theta \sum_i x_i\right),\] 
where $A_{i,j}$s and $\theta$ are parameters of the model. The constant of proportionality, also called the {\em partition function} of the Ising model is $Z=\sum_x N(x)$, which gives $P(x)=N(x)/Z$. The {\em width} of the Ising model is defined as $\max_i \sum_j |A_{i,j}|+\theta$.  In a {\em ferromagnetic} Ising model, each $A_{ij} \geq 0$.

Given two such Ising models, we give an algorithm for additively estimating their total variation distance. We first learn these two Ising models up to total variation distance $\epsilon/8$ using the following learning algorithm given by Klivans and Meka~\cite{KM17}. In fact, it gives a stronger $(1\pm\epsilon)$ multiplicative approximation guarantee for every probability value.

\begin{theorem}[Theorem 7.3 in \cite{KM17}]\label{thm:IsingLearn}
There is an algorithm which, given independent samples from an unknown Ising model  $P$ with width at most $d$, returns parameters $\hat{A}_{i,j}$ and $\hat{\theta}$ such that the Ising model $\hat{P}$ constructed with the latter parameters satisfies $(1-\epsilon)P(x)\le \hat{P}(x) \le (1+\epsilon) P(x)$ for all $x \in \{-1,1\}^n$. This algorithm takes $m=e^{O(d)}\epsilon^{-4}n^8\log ({n/\delta\epsilon})$ samples, $O(mn^2)$ time and succeeds with probability $1-\delta$. 
\end{theorem}

However learning the parameters of an Ising model is not enough to efficiently evaluate the probability at arbitrary points. Naively computing the constant of proportionality $Z$ would take $2^n$ time. For certain classes of Ising models polynomial time algorithms are known which approximates $Z$ up to a $(1\pm\epsilon)$ approximation factor. In particular we use the following approximation algorithm for ferromagnetic\footnote{As pointed out by \cite{SriComm}, Jerrum and Sinclair's result (and hence, our result) extends to the {\em non-uniform external field} setting where there is a $\theta_i$ for each $i$ instead of $\theta_1 = \cdots = \theta_n = \theta$, with the restriction that each $\theta_i \geq 0$.} Ising models due to Jerrum and Sinclair~\cite{JS93}.

\begin{theorem}\label{thm:IsingPartition}
There is an algorithm which given the parameters of a ferromagnetic Ising model distribution $P$, in $O(\epsilon^{-2} n^{17} \log n)$ time returns a number $\hat{Z}$ such that with probability at least 9/10, $(1-\epsilon)Z \le \hat{Z} \le (1+\epsilon)Z$, where $Z$ is the partition function of $P$. 
\end{theorem}

Combining the previous two results with our general distance estimation algorithm, we can now obtain our main result for Ising models which we restate below.
\ignore{
\begin{theorem}
Let $P$ and $Q$ be two unknown ferromagnetic Ising models of width at most $d$ which we can access by samples. Then there is an algorithm which takes $m=O(e^{O(d)}\epsilon^{-4}n^8\log ({n\over \epsilon}))$ samples from $P$ and $Q$, and in $O(mn^2+\epsilon^{-2} n^{17} \log n)$ time\snote{The time complexity turned out more than that stated in the introduction} returns $e$ such that $|e - \dtv(P,Q)|\le\epsilon$ with probability at least 2/3.
\end{theorem}}
\isingmain*
\begin{proof}
We first use~\cref{thm:IsingLearn} to get the parameters for a pair of Ising models $\hat{P}$ and $\hat{Q}$  which are, with probability at least $9/10$, pointwise $(1\pm\epsilon/8)$ approximations to $P$ and $Q$. If $\hat{P}$ or $\hat{Q}$ has any negative pairwise interaction term, then we modify them to zero, thus making $\hat{P}$ and $\hat{Q}$ ferromagnetic.
We claim that since $P$ and $Q$ are ferromagnetic to start with, this can only improve the approximation factor. The reason is that Klivans and Meka, in their proof of \cref{thm:IsingLearn}, show the more general result that for any {\em log-polynomial distribution}, i.e, any distribution $P$ on $\{-1,1\}^n$ where $P(x) \propto \exp(T(x))$ for a bounded-degree polynomial $T$, they can obtain a polynomial $\hat{T}$ with the same degree that satisfies a bound on $\|T-\hat{T}\|_1 = \sum_\alpha |T[\alpha]-\hat{T}[\alpha]|$ where $T[\alpha]$ and $\hat{T}[\alpha]$ are the coefficients of the monomial indexed by $\alpha$. It is clear that if $T[\alpha]\geq 0$, changing $\hat{T}[\alpha]$ to $\max(0,\hat{T}[\alpha])$ can only reduce $\|T-\hat{T}\|_1$. 
\ignore{
Klivans and Meka's algorithm~\cite{KM17} actually gives an additive $\eps$ approximation $\hat{T}$ of the term $T$ inside $\exp()$ where $N(x)=\exp(T)$ to get an $\exp(\epsilon)\simeq (1+\epsilon)$ approximation for $P(x)$. We modify the $\hat{A}_{i,j}<0$ value to $\hat{A}_{i,j}=0$ in $\hat{P}$ and $\hat{Q}$ if there are any one by one. During each modification since other $\hat{A}_{i,j}$ values and $\hat{\theta}$ remain  unchanged, the difference only comes because of the current $\hat{A}_{i,j}$ under modification. Since the ground truth $A_{i,j} \ge 0$ for $P$ and $Q$ whereas the estimate from~\cite{KM17} have $\hat{A}_{i,j}<0$, modifying $\hat{A}_{i,j}=0$ makes $\hat{P}(x)$ closer to the ground truth $P(x)$ for every possibilities of $x\in \{-1,1\}^n$. Thus every such modification goes closer to $T$ than before and we would end up with a ferromagnetic Ising model.
}

Abusing notation for simplicity, henceforth let $\hat{P}$ and $\hat{Q}$ be the distributions after this modification. Let $N_{\hat{P}}(x)$ and $N_{\hat{Q}}(x)$ be the numerators for $\hat{P}$ and $\hat{Q}$ respectively. Then we apply~\cref{thm:IsingPartition} to estimate, with probability $4/5$, the partition functions $\hat{Z}_P$ and $\hat{Z}_Q$ of $\hat{P}$ and $\hat{Q}$ respectively up to a $(1\pm\epsilon/8)$ multiplicative factor. Therefore, $E_P(x)=N_{\hat{P}}(x)/\hat{Z}_P$ and $E_Q(x)=N_{\hat{Q}}(x)/\hat{Z}_Q$ are $(\epsilon/8,\epsilon/4)$-EVAL approximators for $P$ and $Q$ respectively, where the $\epsilon/8$-close distributions are $\hat{P}$ and $\hat{Q}$. 
It follows from~\cref{thm:main} that conditioned on the above, we can estimate $\dtv(P,Q)$ up to an $\epsilon$ additive error with probability at least $9/10$. 
\end{proof}

\subsection{Distance to uniformity}
Next we give an algorithm for estimating the distance between an unknown Ising model and the uniform distribution over $\{-1,1\}^n$. 
\isingsupp*
\begin{proof}
We first learn the ising model using \cref{thm:IsingLearn}. As we noted earlier computing the partition function naively is intractable in general. However computing $N_x/N_z$, the ratio of the probabilities of two items $x,y$ can be computed in $O(n^2)$ time up to $(1\pm\epsilon)$ approximation from \cref{thm:IsingLearn}. Canonne et al.~\cite{DBLP:journals/siamcomp/CanonneRS15} have given an algorithm for computing distance to uniformity from an unknown distribution using sampling and pairwise conditional sampling (PCOND) access to it using $m_1=O(\epsilon^{-19}\log^8 {1\over \epsilon})$ PCOND samples and $m_2=O(\epsilon^{-7}\log^3 {1\over \epsilon})$ samples with probability 2/3 up to a $O(\epsilon)$ additive error. A closer look at their algorithm reveals that all their PCOND accesses are made from a routine called `COMPARE', whose job is to compute the ratio of probabilities of two points $x$ and $z$ with probability $1-\delta$ upto $(1\pm\eta)$-factor using conditional samples. In fact it suffices for their algorithm to correctly compute the ratio if it is in $[1/K,K]$, report `HIGH' if it is in $(K,\infty]$, and `LOW' if it is in $[0,1/K)$ for a parameter $K$. In the case of ising model, assuming success of \cref{thm:IsingLearn} we can replace the routine `COMPARE' by computing $N_x/N_z$ using the parameters of the learnt model upto $(1\pm\epsilon)$ approximation in $O(n^2)$ time with $\delta=0$. Their algorithm makes $m_3=O(\epsilon^{-7}\log^3 {1\over \epsilon})$ calls to `COMPARE'. Using their choices of various parameters our theorem follows.
\end{proof}

\section{Multivariate Gaussians}\label{sec:gaussians}

In this section we give an algorithm for additively estimating the total variation distance between two unknown multidimensional Gaussian distributions. 
For a mean vector $\mu \in \R^n$ and a positive definite covariance matrix $\Sigma \in \R^{n \times n}$, the Gaussian distribution $N(\mu, \Sigma)$ has the pdf:
\begin{equation}
N(\mu, \Sigma; x) = \frac{1}{(2\pi)^{n/2} \sqrt{\det(\Sigma)}} \exp\left(-\frac12(x-\mu)^\top \Sigma^{-1} (x-\mu)\right)
\end{equation}
We use the following folklore learning result for learning the two Gaussians.

\begin{theorem}\label{thm:GaussLearn}
Let $P$ be an $n$-dimensional Gaussian distribution. Let $\hat{\mu} \in \R^n$ and $\hat{\Sigma} \in \R^{n \times n}$ be the empirical mean and the empirical covariance defined by
$O(n^2 \eps^{-2})$ samples from $P$. Then, with probability at least $9/10$, the distribution $\hat{P} = N(\hat{\mu}, \hat{\Sigma})$ satisfies $\dtv(P,\hat{P}) \leq \eps$.
\end{theorem}

We are now ready to prove \cref{thm:gaussians-main} restated below.
\ignore{\begin{theorem}\label{thm:tv-approx}
Let $P$ and $Q$ be two $n$ dimensional Gaussian distributions. There is an algorithm that takes $O(n^2/\epsilon^2)$ samples from $P$ and $Q$ and in $O(n^{\omega}\epsilon^{-2})$ time returns a value $e$ such that $|e-\dtv(P,Q)|\le \epsilon$ with probability 2/3, where $\omega\simeq 2.373$ is the matrix multiplication constant.
\end{theorem}}
\gaussianmain*
\begin{proof}
We first apply \cref{thm:GaussLearn} to obtain $\hat{P}$ and $\hat{Q}$ such that each is within $\eps/4$ distance from $P$ and $Q$ respectively. Since we can evaluate the pdf of $\hat{P}$ and $\hat{Q}$ exactly, they serve as $(\epsilon/4,0)$ \eval-approximators for $P$ and $Q$. Each determinant computation costs $O(n^{\omega})$ time.
Subsequently from (the continuous analog of)~\cref{thm:main}, using $O(\epsilon^{-2})$ samples from $P$ and $O(n^{\omega}\epsilon^{-2})$ time, we can estimate $\dtv(P,Q)$ up to an additive $\epsilon$ error with probability at least $4/5$.
\end{proof}
\begin{remark}
The above time analysis uses the unrealistic real RAM model in which real number computations can be carried out exactly upto infinite precision. However, there are strongly polynomial time algorithms for computing matrix determinant and inverse \cite{LovNotes, Wilk65}, so that even in the more realistic word RAM model, the above algorithm runs in polynomial time.
\end{remark}

As a by-product of our analysis, we also obtain an efficient randomized algorithm to compute the total deviation distance between two gaussians specified by their parameters.
\begin{corollary}\label{cor:tv-approx}
For any two vectors $\mu_1,\mu_2\in \R^n$ and two positive-definite matrices $\Sigma_1,\Sigma_2 \in \R^{n\times n}$, $\dtv(N(\mu_1,\Sigma_1)\allowbreak ,N(\mu_1,\Sigma_1))$ can be estimated up to an additive $\epsilon$ error in $O(n^3\eps^{-2})$ time.
\end{corollary}
\begin{proof}
We again invoke \cref{algo:main}. Since the parameters are already provided, we can readily obtain $(0,0)$-\eval approximators for $N(\mu_1, \Sigma_1)$ and $N(\mu_2,\Sigma_2)$. For \cref{algo:main}, we also need sample access to one of the two distributions. It is well known that if $v \sim N(0,I)$ and $\Sigma = LL^\top$, then $Lv+\mu \sim N(\mu,\Sigma)$; the matrix $L$ can be obtained in $O(n^3)$ time using a Cholesky decomposition. Hence, each sample from $N(\mu_1,\Sigma_1)$ costs $O(n^3)$ time, so that the entire algorithm runs in $O(n^3\eps^{-2})$ time.
\end{proof}

\section{Causal Bayesian Networks under Atomic Interventions}\label{sec:cbn}
We describe Pearl's notion of causality from~\cite{Pearl00}. Central to his formalism is the notion of an {\em intervention}. Given a variable set ${V}$ and a subset ${S} \subset {V}$, an intervention $\cdo({s})$ is the process of fixing the set of variables in ${S}$ to the values ${s}$. If the original distribution on $V$ is $P$, we denote the {\em interventional distribution} as $P_s$, intuitively, the distribution induced on ${V}$ when an external force sets the variables in ${S}$ to ${s}$.

Another important component of Pearl's formalism is that some variables may be hidden (latent). The hidden variables can neither be observed nor be intervened upon.  Let ${V}$ and ${U}$ denote the subsets corresponding to observable and hidden variables respectively. 
Given a directed acyclic graph $H$ on ${V \cup U}$ and a subset ${S} \subseteq ({V \cup U})$, we use ${\Pi}_H({S})$ and $\Pa_H({S})$  to denote the set of all parents and observable parents respectively of ${S}$, excluding ${S}$, in $H$. When the graph $H$ is clear, we may omit the subscript. 

\begin{definition}[Causal Bayesian Network] \label{def:causal Bayesnet}
A {\em (semi-Markovian) causal Bayesian  network (\cbn)} on variables $X_1, \dots, X_n$ is a collection of interventional distributions defined by a tuple $\langle {V}, {U}, G,$ $\{\Pr[X_i \mid x_{\Pi(i)}] : i \in {V}, x_{\Pi(i)} \in \Sigma^{|{\Pi}(i)|}\}, \Pr[X_U]\}\rangle$, where (i)  $G$ is a directed acyclic graph on ${V} \cup {U} = [n]$, (ii) $\Pr[X_i \mid x_{\Pi(i)}]$ is the conditional probability distribution of $X_i$ given that its parents $X_{{\Pi}(i)}$ take the values $x_{{\Pi}(i)}$, and (iii)  $\Pr[X_U]$ is the distribution of the hidden variables $\{X_i : i \in U\}$.

A \cbn\ 
$\cP = \langle {V}, {U}, G,$ $\{\Pr[X_i \mid x_{\Pi(i)}]: i \in V, x_{\Pi(i)} \in \Sigma^{|\Pi(i)|}\}, \Pr[X_U]\rangle$ 
defines a unique interventional distribution $P_s$ for every subset ${S} \subseteq {V}$ (including ${S} = \emptyset$) and assignment ${s} \in \Sigma^{|{S}|}$, as follows. For all ${x} \in \Sigma^{|{V}|}$:
$$P_s({x}) = 
\begin{cases}
\sum_{{u}} \prod_{i \in {V}\setminus {S}} \Pr[x_i \mid x_{\pi(i)}] \cdot \Pr[X_U = u] & \text{if }{x} \text{ is consistent with }{s}\\
0 & \text{ otherwise.}
\end{cases}
$$
We use $P$ to denote the observational distribution ($S = \emptyset$).
$G$ is said to be the {\em causal graph} corresponding to the {\cbn} $\cP$.
\end{definition}
It is standard in the causality literature~\cite{10.5555/2073876.2073938,VERMA199069,10.5555/3327546.3327616} to assume that each variable in $U$ is a source node with exactly two children from $V$, since there is a known algorithm~\cite{10.5555/2073876.2073938,VERMA199069} which converts a general causal graph into such graphs. Given such a causal graph, we remove every source node $Z$ from $G$ and put a {\em bidirected} edge between its two observable children $X_1$ and $X_2$. We end up with an Acyclic Directed Mixed Graph (ADMG) graph $G$, having vertex set $V$ and having edge set $E^\to \cup E^\leftrightarrow$ where $E^\to$ are the directed edges and $E^\leftrightarrow$ are the bidirected edges. The {\em in-degree} of $G$ is the maximum number of directed edges coming into any vertex in $V$. A {\em c-component} refers to any maximal subset of $V$ which is interconnected by bidirected edges. Then $V$ gets partitioned into c-components: $S_1, S_2,\dots,S_\ell$. \cref{fig:admg} shows an example.

\begin{figure}
\centering
\includegraphics{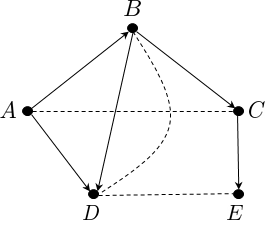}
\caption{An acyclic directed mixed graph (ADMG) where the bidirected edges are depicted as dashed. The in-degree of the graph is 2. The c-components are $\{A,C\}$ and $\{B,D,E\}$.}
\label{fig:admg}
\end{figure}

Throughout this section, we focus on {\em atomic} interventions, i.e. interventions on a single variable.
Let $A \in V$ correspond to this variable. Without loss of generality, suppose $A \in S_1$. 
Tian and Pearl \cite{TP02b} showed that in an ADMG $G$ as above, $P_a$ can be completely determined from $P$ for all $a \in \Sigma$ iff the following condition holds.

\begin{assumption}[Identifiability wrt $A$]
There does not exist a path of bidirected edges between $A$ and any child of $A$. Equivalently, no child of $A$ belongs to $S_1$.
\label{ass:id}
\end{assumption}

Recently algorithms and sample complexity bounds for learning and sampling from identifiable atomic interventional distributions were given in~\cite{BGKMV20} under the following additional assumption.
For $S\subseteq V$, let $\Pa^+(S)=S\cup \Pa(S)$.

\begin{assumption}[$\alpha$-strong positivity wrt $A$]
Suppose $A$ lies in the c-component ${S}_1$, and let $Z = \Pa^+(S_1)$. For every assignment $z$ to $Z$, $P(Z=z) > \alpha$. 
\label{ass:bal}
\end{assumption}

We state the two main results of \cite{BGKMV20}, which given sampling access to the observational distribution $P$ of an unknown causal Bayesian network on a known ADMG return an $(\epsilon,0)$-\eval approximator and an approximate generator for $P_a$.
For the two results below, suppose the CBN $\cP$ satisfies identifiablity (\cref{ass:id}) and $\alpha$-strong positivity (\cref{ass:bal}) with respect to a variable $A \in {V}$ . Let $d$ denote the maximum in-degree of the graph $G$ and $k$ denote the size of its largest c-component. 

\begin{theorem}[\eval approximator~\cite{BGKMV20}]\label{thm:doEVAL}
For any intervention $a$ to $A$ and parameter $\eps \in (0,1)$, there is an algorithm that takes $m = \tilde{O}\left(\frac{|\Sigma|^{5kd}n}{\alpha\epsilon^2}\right)$ samples from $P$, and in $O(mn|\Sigma|^{2kd})$ time, returns a circuit $E_{P,a}$. With probability at least $2/3$,  the circuit $E_{P,a}$ implements an $(\eps,0)$-\eval approximator for $P_a$, and it  runs in $O(n)$ time for all inputs.
\end{theorem}

\begin{theorem}[Generator~\cite{BGKMV20}]\label{thm:doSAMP}
For any intervention $a$ to $A$ and parameter $\eps \in (0,1)$, there is an algorithm that takes $m = \tilde{O}\left(\frac{|\Sigma|^{5kd}n}{\alpha\epsilon^2}\right)$ samples from $P$, and in $O(mn|\Sigma|^{2kd})$ time, returns a probabilistic circuit $G_{P,a}$ that generates samples of a distribution $\tilde{P}_a$ satisfying $\dtv(P_a,\tilde{P}_a) \leq \eps$. On each call, the circuit takes $O(n |\Sigma|^{2kd}\eps^{-1}\log \delta^{-1})$ time and outputs a sample of $\tilde{P}_a$ with probability at least $1-\delta$. 
\ignore{
takes $m = \tilde{O}\left(\frac{|\Sigma|^{5kd}n^2}{\alpha\epsilon^2}\right)$ samples of $P(V)$, and in $O(mn|\Sigma|^{2kd})$ time produces a distribution $\hat{P}_{\samp}$ such that $\dtv(P(X_{V\setminus \{i^*\}}\mid\cdo(X_{i^*}=\sigma)),\hat{P}_{\samp}) \le \epsilon$ with probability at least $1-\delta$. $\hat{P}_{\samp}$ returns a set of $t$ samples in $O(nt|\Sigma|^{2kd}\log {t\over \gamma}/\epsilon)$ time with probability at least $1-\gamma$ for any $t,\gamma$.}
\end{theorem}

We give a distance approximation algorithm for identifiable atomic interventional distributions using the above two results and \cref{thm:main}.

\begin{theorem}[Formal version of \cref{thm:cbninf}]
\ignore{Let $G_1$ and $G_2$ be two known ADMG's on a common observable set $V$ containing a special node $A$. Suppose $G_1$ and $G_2$ have maximum in-degree $\leq d$ and largest c-components of size $\leq k$.
Suppose both $G_1$ and $G_2$ satisfy \cref{ass:id}.

Fix any $a \in \Sigma$. Define $\cD_1$ and $\cD_2$ to be the families of interventional distributions $P_a$ and $Q_a$ obtained from CBN's $\cP$ and $\cQ$  over $G_1$ and $G_2$ respectively satisfying \cref{ass:bal} wrt $A$. Then there is a distance approximation algorithm for $(\cD_1, \cD_2)$ that takes $m = \tilde{O}\left(\frac{|\Sigma|^{5kd}n^2}{\alpha\epsilon^2}\right)$ samples from $P$ and $Q$, runs in time $\tilde{O}(mn|\Sigma|^{2kd}+n|\Sigma|^{2kd}\epsilon^{-3})$}

Suppose $\cP, \cQ$ are two unknown CBN's on two known ADMGs $G_1$ and $G_2$ on a common observable set $V$ both satisfying \cref{ass:id} and \cref{ass:bal} wrt a special vertex $A$. Let $d$ denote the maximum in-degree, and $k$ denote the size of the largest c-component of $G_1$ and $G_2$. 

Then there is an algorithm which for any $a \in \Sigma$ and parameter $\eps \in (0,1)$, takes $m=\tilde{O}\left(\frac{|\Sigma|^{5kd}n}{\alpha\epsilon^2}\right)$ samples from $P$ and $Q$, runs in time $\tilde{O}(mn|\Sigma|^{2kd}+n|\Sigma|^{2kd}\epsilon^{-3})$ and returns a value $e$ such that
$|e-\dtv(P_a,Q_a)|\le \epsilon$ with probability at least 2/3.
\end{theorem}
\begin{proof}
We first invoke \cref{thm:doSAMP} to 
obtain the generators for distributions $\tilde{P}_a$ and $\tilde{Q}_a$ that are $\epsilon/10$ close to the two interventional distributions $P_a$ and $Q_a$ respectively in $\dtv$. By triangle inequality, it suffices to estimate $\dtv(\tilde{P}_a,\tilde{Q}_a)$ up to an additive $4\epsilon/5$ error. Next we 
invoke \cref{thm:doEVAL} to obtain circuits $E_{P,a}$ and $E_{\cQ,a}$ that each implement  $(\epsilon/10,0)$-\eval approximators for  the two interventional distributions $P_a$ and $Q_a$ respectively. Let $\hat{P}_a$ and $\hat{Q}_a$ denote the two distributions that $E_{P,a}$ and $E_{Q,a}$ respectively compute evaluations of.
Using the triangle inequality, $\dtv(\tilde{P}_a,\hat{P}_a) \le\epsilon/5$ and $\dtv(\tilde{Q}_a,\hat{Q}_a) \le\epsilon/5$. Thus $E_{P,a}$ and $E_{Q,a}$ are $(\epsilon/5,0)$-\eval approximators for $\tilde{P}_a$ and $\tilde{Q}_a$ respectively. From \cref{thm:main}, we need $O(\epsilon^{-2})$ samples from $\tilde{P}_a$ and $O(\eps^{-2})$ calls to $E_{P,a}$ and $E_{Q,a}$  to estimate $\dtv(\tilde{P}_a,\tilde{Q}_a)$ up to an additive $4\epsilon/5$ error.
\end{proof}
\section{Improving Success of Learning Algorithms Using Distance Estimation}\label{sec:boost}
In this section we give a general algorithm for improving the success probability of learning certain families of distributions. Specifically, let $\mathcal{D}$ be a family of distributions for which we have a learning algorithm $\mathcal{A}$ in $\dtv$ distance $\epsilon$ that succeeds with probability 3/4. Suppose there is also a distance approximation algorithm $\mathcal{B}$ for $\mathcal{D}$. The algorithm presented below, which uses $\mathcal{A}$ and $\mathcal{B}$, learns an unknown distribution from $\mathcal{D}$ with probability at least $(1-\delta)$.

\begin{algorithm}[H]
\KwData{Samples from an unknown distribution $P$}
\KwResult{A distribution $\hat{P}$ such that $\dtv(P,\hat{P})\le \epsilon$ with probability $1-\delta$}
\For{$0\le i \le R=O(\log {1\over \delta})$}{
$P_i \gets$ Run $\mathcal{A}$ on samples from $P$ to get a learnt distribution\;
$count_i \gets 0$\;
}
\For{every unordered pair $0\le i<j\le R$}{
$d_{ij} \gets$ Estimate distance between $P_i$ and $P_j$ up to additive error $\epsilon$ using $\mathcal{B}$\;
\If{$d_{ij} \le 3\epsilon$}{
$count_i \gets count_i+1$\;
$count_j \gets count_j+1$\;
}
}
$i^*=\arg \max_i count_i$\;
\Return $P_{i^*}$\;
\caption{High probability distribution learning}
\label{algo:boost}
\end{algorithm}

\boostmain*
\ignore{
\begin{theorem}\label{thm:boost}
Let $\mathcal{D}$ be a family of distributions. Suppose there is a learning algorithm $\mathcal{A}$ which for any $P\in \mathcal{D}$ takes $m_{\mathcal{A}}(\epsilon)$ samples from $P$ 
and in time $t_{\mathcal{A}}(\epsilon)$ outputs a distribution $P_1$ such that $\dtv(P,P_1)\le \epsilon$ with probability at least 3/4. Suppose there is a distance approximation 
algorithm $\mathcal{B}$ for $\mathcal{D}$ that given any two completely specified distributions $P_1$ and $P_2$ estimates $\dtv(P_1,P_2)$ up to an additive error $\epsilon$ in 
$t_{\mathcal{B}}(\epsilon,\delta)$ time with probability at least $(1-\delta)$. Then Algorithm~\ref{algo:boost} uses $\mathcal{A}$ and $\mathcal{B}$, takes $O(m_{\mathcal{A}}
(\epsilon/4)\log {1\over \delta})$ samples from $P$, runs in $O(t_{\mathcal{A}}(\epsilon/4)\log {1\over \delta}+t_{\mathcal{B}}(\epsilon/4,{\delta\over 210000 \log^2 {2\over \delta}})\log^2 {1\over \delta})$ time and returns a distribution $\hat{P}$ such that $\dtv(P,\hat{P})\le \epsilon$ with probability at least $1-\delta$.  
\end{theorem}}
\begin{proof}
The boosting algorithm is given in Algorithm~\ref{algo:boost}. We take $R=324\log {2\over \delta}$  repetitions of $\mathcal{A}$ to get the distributions $P_i$s. From Chernoff's bound at least $2R/3$ distributions (successful) satisfy $\dtv(P_i,P)\le \epsilon$ with probability at least $1-\delta/2$, which we condition on henceforth. These successful distributions have pairwise distance at most $2\epsilon$. Conditioned on the ${R \choose 2}$ calls to $\mathcal{B}$ succeeding, the pairwise distances between the successful distributions are at most $3\epsilon$. Hence every successful $i$ has its count value at least $2R/3-1$. This means $i^*$, which has the maximum count value $(\ge 2R/3-1)$ must intersect at least one successful $i'$ such that $\dtv(P_{i^*},P_{i'})\le 3\epsilon$. By triangle inequality we get $\dtv(P_{i^*},P)\le 4\epsilon$.

It suffices for each call to $\mathcal{B}$ succeed with probability at least ${\delta\over 2R^2}$.
\end{proof}

Assuming black-box access to $\mathcal{A}$ $O(m_{\mathcal{A}}\log {1\over \delta})$ samples are needed in the worst case to learn with $1-\delta$ probability since otherwise all the $o(\log {1\over \delta})$ repetitions may fail. We can apply the above algorithm to improve the success probability of learning bayesian networks on a given graph with small indegree and multidimensional Gaussians.
\bibliographystyle{alpha}
\bibliography{reflist}

\end{document}